\documentclass[journal,comsoc, onecolumn]{IEEEtran}

\usepackage[T1]{fontenc}
\usepackage{color}
\usepackage{graphicx}
\usepackage{amsthm}
\newtheorem{proposition}{Proposition}
\newtheorem{corollary}{Corollary}
\newtheorem{definition}{Definition}
\newtheorem{remark}{Remark}
\newtheorem{lemma}{Lemma}
\newtheorem{theorem}{Theorem}
\newtheorem{example}{Example}
\usepackage{epstopdf}
\newcommand{\underrel}[2]{\mathrel{\mathop{#2}\limits_{#1}}}
\newcommand{\eqv}{ \underrel{n}{\sim} }
\newcommand{\eqn}[1]{ \underrel{#1}{\sim} }
\newcommand{\argmax}[1]{\mbox{arg }\max_{#1}}
\newcommand{\F}{\mathbb{F}}
\newcommand{\R}{\mathbb{R}}
\newcommand{\Z}{\mathbb{Z}}
\newcommand{\X}{\mathcal{X}}
\newcommand{\Perror}{P_{\mbox{\small{error}}}}

\ifCLASSINFOpdf
\else
\fi

\usepackage{amsmath}
\usepackage{bm}
\usepackage{mathrsfs}

\begin{document}
\title{On Equivalence of Binary Asymmetric Channels regarding the Maximum Likelihood Decoding.} \author{Claudio~Qureshi, Sueli~I.~R.~Costa, Christiane~B.~Rodrigues  and Marcelo~Firer 
\thanks{The authors are with the Institute of Mathematics, Statistics and Computing Science of the University of Campinas, SP , Brazil (emails: cqureshi@ime.unicamp.br, sueli@ime.unicamp.br, chrismmor@gmail.com, mfirer@ime.unicamp.br).}
\thanks{}}


\maketitle

\begin{abstract}
We study the problem of characterizing when two memoryless binary asymmetric channels, described by their transition probabilities $(p,q)$ and $(p',q')$, are equivalent from the point of view of maximum likelihood decoding (MLD) when restricted to $n$-block binary codes. This equivalence of channels induces a partition (depending on $n$) on the space of parameters $(p,q)$ into regions associated with the equivalence classes. Explicit expressions for describing these regions, their number and areas are derived. Some perspectives of applications of our results to decoding problems are also presented.

\end{abstract}


\IEEEpeerreviewmaketitle

\section{Introduction}\label{SectionIntroduction}

The binary asymmetric channel $BAC(p,q)$ is the discrete memoryless channel with binary alphabet and transition probability given by  $\mbox{Pr}(1|0)=p$, $\mbox{Pr}(0|0)=1-p$, $\mbox{Pr}(0|1)=q$ and $\mbox{Pr}(1|1)=1-q$, where $\mbox{Pr}(x|y)$ denotes the probability of receiving $x$ if $y$ was sent. Without loss of generality we will assume $p\leq q$ and $p+q<1$ (see discussion in Section \ref{SectionFurtherRemarks}) and denote \emph{the space of parameters} by $\mathcal{T}=\{(p,q)\in [0,1]^2: p\leq q, p+q<1\}$ which we refer to as the \emph{fundamental triangle}.\\

The interest in binary asymmetric channels has increased due to applications in flash memories \cite{CGOZ99,CSBB10,KBE10,YSVW11} and as models in other areas, such as neuroscience \cite{CIMRW13}. Most of the work developed on binary asymmetric channels has focused on coding properties \cite{CR79,MR80} and on the design of codes with given properties, either for general asymmetric channels \cite{Berger61,CR79,GD12,Klove81,VT65,Weber89} or for the $Z$-channel \cite{LK14,MGK98,PBC15}. We stress that most of these works consider the decoding criterion determined by the asymmetric metric introduced in \cite{CR79}, or some variant of it. Despite of all its advantages, the asymmetric metric is not matched to any binary asymmetric channel, that is, it cannot be used to perform maximum likelihood decoding (MLD).

In this work we study the binary asymmetric channels from the point of view of maximum likelihood decoding. Namely, two memoryless binary asymmetric channels $W_1$ and $W_2$ are $n$-equivalent if MLD is the same for both channels, for every possible code $C\subseteq \F_2^n$ (a precise statement is given in Definition \ref{Def-n-Equivalence}). The most studied instances of these channels are the binary symmetric channels (BSCs) corresponding to $p=q$, and the $Z$-channels corresponding to $p=0$. For these channels the problem we study here is trivial since two channels $W_{1}$ and $W_{2}$ which are either both symmetric channels or both $Z$-channels are always
$n$-equivalent, for any positive $n$. In this sense, we could say that there is a unique BSC and a unique $Z$-channel from the MLD point of view. This is not the case of general binary asymmetric channels and to study this equivalence relation is the focus of this work.


The knowledge of the equivalence classes of binary asymmetric channels may be useful for two purposes.

\begin{enumerate}  
\item To perform MLD on a memoryless binary channel it is necessary to know the transition probabilities $p=\Pr(1|0)$ and $q=\Pr(0|1)$. The more precise the measurement of $(p,q)$, the less the risk of mismatching the channel. Let us suppose that after a number of experiments, we obtain an approximation $(p',q')$ for the real transition probabilities $(p,q)$ with an error of at most $\varepsilon>0$ (i.e. such that $|p-p_0|<\varepsilon$ and $|q-q_0|<\varepsilon$). We prove that the larger the block length of a code, the more is the risk of mismatching\footnote{To mismatch a channel $W$ means to use a decoding criteria different from the MLD with respect to the transition probabilities of $W$. For general references on mismatched decoding see \cite{GLT00,LN98}.} the channel. In fact, there is a maximum block length $N$, which can be calculated explicitly from the results developed in this paper, such that there is no risk of mismatching when codes of block length at most $N$ are considered (see Example \ref{ExampleMismatching1}). A dual situation occurs when we have an upper bound $N$ on the block length of the codes to be used in a given BAC. In this case the probability of mismatching is reduced when we increase the precision in the measurement of the transition probabilities $(p,q)$. It is possible to find a maximum admissible error $\varepsilon$ such that if the estimated transition probabilities $(p',q')$ verify $|p-p'|<\varepsilon$ and $|q-q'|<\varepsilon$, there is no risk of mismatching (see Example \ref{ExampleMismatching2}).
\item A very relevant measure of the performance of a code is the error probability of the encoding-decoding process. Given a memoryless channel $W:\X \rightarrow \X$ and a code $C\subseteq \X^n$, we consider the ML decoder with input $x\in \X^n$ and output the codeword $c\in C$ such that $\Pr_{W}(x|c)>\Pr_{W}(x|c')$ for all $c'\in C, c'\neq c$, if such codeword exists. Otherwise, in the case there are different codewords $c$ that maximize $\Pr_{W}(x|c)$, it returns a 'FAIL' message. We will refer to this ML decoder as the \emph{standard ML decoder}. For this decoder, the \emph{error probability of the code $C$ within the channel $W$} is given by $$ \Perror(C,W) = 1-\frac{1}{|C|}\sum_{c\in C}\ \sum_{x\in V_{(C,W)}(c)} \mbox{Pr}_{W}(x|c)$$ where $V_{(C,W)}(c)=\{x \in \X^n: \Pr_{W}(x|c)> \Pr_{W}(x|c'),\ \forall c' \in C, c'\neq c\}$ is the \emph{probabilistic Voronoi region} of $c$ (depending on $W$ and $C$). Here we are assuming the messages to be equiprobable. For the particular case of a BAC, the error probabilities $\Pr_{W}(y|c)$ are easy to compute, since there are closed formulas for them. The difficult part of computing $\Perror(C,W)$ is determining the probabilistic Voronoi regions. We should mention that, since it is defined depending on $W$, we actually have an infinite number of instances, which a priori should be computed for any given $W$ (and the finite number of possible codes $C$). Our definition of equivalence of channels actually says that two channels $W_{1}$ and $W_{2}$ are $n$-equivalent if $V_{(C,W_1)}(c) = V_{(C,W_2)}(c)$ for every $C\subseteq \X^n$ and every $c\in C$. It follows that, knowing the equivalence classes reduces the problem from infinitely many instances of binary asymmetric channels to a finite number of equivalence classes of such channels (depending on the length $n$ of the block codes).
\end{enumerate}

We may considered another ML decoder different than the standard ML decoder if instead of return a 'FAIL' message in case of ambiguity (i.e. several codewords $c$ maximizing $\Pr_W(x|c)$), it returns a codeword $c$ maximizing $\Pr_W(x|c)$, chosen uniformly at random. This (probabilistic) ML decoder will be called the \emph{uniform ML decoder}. The definition of equivalence of channels proposed in this paper contemplates both decoders, the standard and the uniform ML decoder, as it is shown in the next section.

This paper is organized as follows: In Section \ref{SectionMLDequivalence} we introduce an equivalence relation between channels in such a way that equivalent channels determine equal decoding criteria when MLD is considered and discuss some properties of this relation. We consider for each (fixed) $n\geq 2$ the above equivalence relation restricted to the BACs. In Section \ref{SectionDecisionCriteria} we introduce the BAC-function, the key to describe the regions determined by the equivalence relation in the parameter space (i.e. in the fundamental triangle). The number of such regions is provided for every $n\geq 2$. In Section \ref{SectionFurtherRemarks} we discuss some properties of the BAC-function and the areas of the regions determined by its level curves, which are related to the probability of a random choice of $(p,q)$ to produce a channel $n$-equivalent to a given BAC.

\section{$n$-equivalence of channels}\label{SectionMLDequivalence}

We start by introducing an equivalence relation (depending on $n\in \mathbb{N}$) that characterizes when two memoryless channels with input and output alphabet $\X$ determine the same ML decision for every $n$-block codes (i.e. subsets $C\subseteq \X^n$).\\

Let $W$ be a memoryless channel with input and output alphabet $\mathcal{X}$ and $n\in \mathbb{N}$. For $x=(x_1,\ldots,x_n),y=(y_1,\ldots,y_n) \in \mathcal{X}^{n}$ we denote by $\mbox{Pr}_{W}(x|y):=\prod_{i=1}^{n}\mbox{Pr}_{W}(x_i|y_i)$ the probability of receiving $x$ if $y$ was sent through the channel $W$. When $W=BAC(p,q)$ we denote this probability by $\mbox{Pr}_{(p,q)}(x|y)$ and if the parameters $(p,q)$ are clear from the context or irrelevant, we denote it by $\mbox{Pr}(x|y)$. Motivated by the definition of matching given in \cite[Definition 1]{FW16}, we introduce the following equivalent relation between memoryless channels.

\begin{definition}\label{Def-n-Equivalence}
Let $W_1,W_2 : \X \rightarrow \X$ be two memoryless channels and $n$ be a positive integer.. We say that $W_1$ and $W_2$ are $n$-equivalent (denoted by $W_1 \eqv W_2$) if for every $n$-block code $C\subseteq \X^{n}$ and every word $x\in \X^{n}$, we have $$ \mbox{arg }\max_{c\in C} \mbox{Pr}_{W_1}(x|c) = \mbox{arg }\max_{c\in C} \mbox{Pr}_{W_2}(x|c) $$ where $\mbox{arg }\max_{c\in C} \mbox{Pr}_{W}(x|c) = \{c\in C: \Pr_W(x|c)\geq \Pr_W(x|c'),\ \forall c' \in C\}$. The channels $W_1$ and $W_2$ are $\infty$-equivalent (denoted by $W_1 \eqn{\infty} W_2$) if they are $n$-equivalent for every $n\geq 1$.
\end{definition}

Let $\mbox{sdec}_{W}$ and $\mbox{udec}_{W}$ denote the standard and uniform ML decoder introduced in Section \ref{SectionIntroduction}, with respect to a memoryless channel $W:\X \rightarrow \X$ and let $W_1,W_2 : \X \rightarrow \X$ be two memoryless channels. The equality $\mbox{sdec}_{W_1}=\mbox{sdec}_{W_2}$  for $n$-block codes means $\mbox{sdec}_{W_1}(C,x) = \mbox{sdec}_{W_2}(C,x)$ for every code $C\subseteq \X^n$ and every $x\in \X^n$, and $\mbox{udec}_{W_1}=\mbox{udec}_{W_2}$ for $n$-block codes means the equality of the probabilities $\Pr\left(\mbox{udec}_{W_1}(C,x)=c\right) = \Pr\left(\mbox{udec}_{W_2}(C,x)=c\right)$ for every code $C\subseteq \X^n$, every $x\in \X^n$ and every $c\in C$. The following theorem, whose proof is given in the Appendix, establishes a relation between the $n$-equivalence of channels and the ML decoders mentioned above.

\begin{theorem}\label{TheoChannelEqANDdecoders}
Let $W_1,W_2 : \X \rightarrow \X$ be two memoryless channels and $n\geq 2$. The following assertions are equivalent.
\begin{itemize}
\item[i)] $W_1$ and $W_2$ are $n$-equivalent.
\item[ii)] The standard ML decoders $\mbox{sdec}_{W_1}$ and $\mbox{sdec}_{W_2}$ are the same for $n$-block codes. 
\item[iii)] The uniform ML decoders $\mbox{udec}_{W_1}$ and $\mbox{udec}_{W_2}$ are the same for $n$-block codes. 
\end{itemize}
\end{theorem}

Consider an order in the $\ell$-ary alphabet $\X$ (for $\X =\F_2$ we assume $0<1$) and the lexicographical order in $\X^n$ for every $n\geq 1$. With a memoryless channel $W: \mathcal{X}\rightarrow \mathcal{X}$ we associate an $\ell^n\times \ell^n$ real matrix $M_n(W)$ whose $ij$-entry is given by
$\mbox{Pr}_{W}(x|y)=\prod_{i=1}^{n}\mbox{Pr}_{W}(x_i|y_i)$, where $x=(x_1,\ldots,x_n)$ and $y=(y_1,\ldots, y_n)$ are the $i$-th and $j$-th elements of $\mathcal{X}^{n}$ respectively. We refer to this matrix as the \emph{transition matrix of order $n$ of $W$}, or simply as the \emph{$n$-transition matrix}. Note that the $1$-transition matrix is the usual transition matrix of the channel. For $W=\mbox{BAC}(p,q)$, the corresponding $n$-transition matrix is denote by $M_n(p,q)$. This is a $2^n \times 2^n$ real matrix and it is not difficult to see that if we consider the binary expansion of $i= \sum_{k=1}^n x_k 2^{n-k}$ and $j=\sum_{k=1}^n y_k 2^{n-k}$ ($0\leq i,j<2^n$), the $ij$-element of $M_n(p,q)$ is given by $m_{ij}= \Pr_{(p,q)}(x|y) = \prod_{k=1}^n \Pr_{(p,q)}(x_k|y_k)$. For example, the matrix $M_5(p,q)$ is a $32 \times 32$ matrix whose entry in the row $9=01001_2$ and column $29=11101_2$ is $m_{9,29}=\mbox{Pr}(01001|11101)= \mbox{Pr}(0|0)\mbox{Pr}(0|1)^2\mbox{Pr}(1|1)^2 = (1-p)q^2(1-q)^2$.\\

%
%


When considering $n$-equivalence, we need to compare the entries in each line of the $n$-transition matrix and it is useful to replace it by a simpler matrix. In \cite{DF16}, the authors substituted a matrix $M$ by a matrix $\tilde{M}$ with entries $\tilde{M}_{ij}=k$ if $M_{ij}$ is the $k$-th largest element (allowing ties) of the $j$-th column. In order to determine whether two channels are equivalent we adopt a different, but similar substitution. 

\begin{definition}
Let $\mathcal{M}_{N}(\R)$ denote the set of $N\times N$ real matrices and $M\in \mathcal{M}_{N}(\R)$. The \emph{ordered form} of $M$ is the integer matrix $M^{*} \in \mathcal{M}_{N}(\Z)$ such that $M^{*}_{ij}=\# \{k: 1\leq k \leq N, M_{ik}<M_{ij}\}$.
\end{definition}

The next proposition characterizes the $n$-equivalence of channels in terms of the ordered form of their $n$-transition matrices (the proof of this result is given in the Appendix).


\begin{proposition}\label{PropOrderedMatrixForm}
Let $W_1, W_2:\mathcal{X} \rightarrow \mathcal{X}$ be two memoryless channels and $M_1$ and $M_2$ be their corresponding $n$-transition matrices. The following assertions are equivalent.
\begin{enumerate}
\item[i)] The channels $W_1$ and $W_2$ are $n$-equivalent.
\item[ii)] $\Pr_{W_1}(x|y)\leq \Pr_{W_1}(x|z) \Leftrightarrow\Pr_{W_2}(x|y)\leq \Pr_{W_2}(x|z)$ for all $x,y,z \in \mathcal{X}^{n}$.
\item[iii)] $M_1^{*}=M_2^{*}$.
\end{enumerate}
\end{proposition}

\begin{remark}\label{RemarkN+1impliesN}
If $W:\X \rightarrow \X$ is a memoryless channel such that there exists $\xi \in \X$ such that $\Pr_{W}(\xi | \xi)>0$, then the map $\X^n \rightarrow \X^{n+1}$ given by $x\mapsto x^{*}:= (x,\xi)$ verifies that $\Pr_{W}(x^{*}|y^{*})\leq \Pr_{W}(x^{*}|z^{*}) \Leftrightarrow \Pr_{W}({x}|{y})\leq \Pr_{W}({x}|{z})$ for all $x,y,z \in \X^n$.  As a consequence of this fact, if $W_1,W_2 : \X \rightarrow \X$ are two memoryless channels such that there exists $\xi \in \X$ such that $\Pr_{W_1}(\xi | \xi)>0$ and $\Pr_{W_2}(\xi|\xi)>0$ (this is always the case if $W_1$ and $W_2$ are BACs) we have the following implications:
\begin{itemize}
\item $W_1 \eqn{n+1} W_2$ implies $W_1 \eqn{n} W_2$;
\item $\mbox{sdec}_{W_1}=\mbox{sdec}_{W_2}$ for $n+1$-block codes implies  $\mbox{sdec}_{W_1}=\mbox{sdec}_{W_2}$ for $n$-block codes;
\item $\mbox{udec}_{W_1}=\mbox{udec}_{W_2}$ for $n+1$-block codes implies  $\mbox{sdec}_{W_1}=\mbox{sdec}_{W_2}$ for $n$-block codes.
\end{itemize}
\end{remark}

In what follows we assume the channels are binary ($\X = \F_2$). Let $n$ be a fixed positive integer or infinite. The $n$-equivalence for BACs induces an equivalence relation on the fundamental triangle $\mathcal{T}$: $(p,q)\underset{n}{\sim} (p',q')$ if and only if $BAC(p,q)\eqv BAC(p',q') $. We denote by $\Delta_n = \mathcal{T} /\underset{n}{\sim}$, the set of equivalence classes and by $\pi_n :\mathcal{T}\rightarrow  \Delta_n$ the projection that associates $(p,q)$ to its equivalence class (i.e. $\pi_n(p,q)=\{(p',q') \in \mathcal{T}:(p',q')\eqv (p,q)\}$). The main result of this paper is a complete description of the quotient set $\Delta_n$.

\begin{definition}
A decision criterion of order $n$ for the BACs is an equivalence class $\mathscr{A}\in\Delta_n$ (in particular $\mathscr{A}$ is a subset of $\mathcal{T}$). An $n$-\emph{stable} decision criterion $\mathscr{A}$ is a decision criterion of order $n$ for the BACs which is an open set of $\mathcal{T}$ and an $n$-\emph{unstable} decision criterion is a decision criterion of order $n$ for the BACs with no interior points.
\end{definition}

\begin{remark}
If $\mathscr{A}$ is a $n$-stable decision criterion for the BACs and $(p,q)\in\mathscr{A}$, then maximum likelihood decoding on $BAC(p,q)$ restricted to $n$-block codes, remains the same under small perturbation of the parameter $(p,q)$. 
\end{remark}

\begin{definition}
A point $(p,q)\in \mathcal{T}$ is $n$-\emph{stable} if $(p,q)$ is an interior point of $\pi_n((p,q))$ and $n$-\emph{unstable} otherwise. The $n$-stable region (denoted by $\mathcal{R}_{n}^{st}$) is the set of all $n$-stable points and the $n$-unstable region (denoted by $\mathcal{R}_{n}^{un}$) is the set of all $n$-unstable points. 
\end{definition}  

We will prove later that every decision criterion is either stable or unstable, that is, if a criterion contain an $n$-stable point then all its points are $n$-stable. We conclude this section by discussing how the parameter space $\mathcal{T}$ decomposes into different $n$-equivalence classes for $n\leq 5$.\\

The $1$-transition matrix $M_1(p,q)$ of $\mbox{BAC}(p,q)$ is given by $\left( \begin{array}{cc} 1-p & q \\ p & 1-q \end{array}  \right)$, thus its ordered form $M_{1}(p,q)^{*} = \left( \begin{array}{cc} 1 & 0 \\ 0 & 1 \end{array}  \right) $ does not depend on $(p,q)$ and we have only one criterion, which is stable.\\

The $2$-transition matrix $M_2(p,q)$ of $\mbox{BAC}(p,q)$ is given by 
$$ \left(\begin{array}{cccc} (1-p)^2 & (1-p)q & (1-p)q & q^2 \\ (1-p)p & (1-p)(1-q) & pq & q(1-q) \\ (1-p)p & pq & (1-p)(1-q) & q(1-q) \\ p^2 & p(1-q) & p(1-q) & (1-q)^2 \end{array}  \right).$$ In this case the ordered form depends on $(p,q)$ in the following way:

\begin{itemize}
\item If $(p,q)$ is an interior point of $\mathcal{T}$ the ordered form is given by $M_2(p,q)^{*} = \left( \begin{array}{cccc} 3 & 1 & 1 & 0 \\ 1 & 3 & 0 & 2 \\ 1 & 0 & 3 & 2 \\ 0 & 1 & 1 & 3 \end{array} \right)$;
\item If $p=0$ (and hence $q>0$, since $(0,0)\not\in \mathcal{T}$) we obtain the ordered form 

$M_2(p,q)^{*} = \left( \begin{array}{cccc} 3 & 1 & 1 & 0 \\ 0 & 3 & 0 & 2 \\ 0 & 0 & 3 & 2 \\ 0 & 0 & 0 & 3 \end{array} \right)$; 
\item If $p=q$ we obtain $M_2(p,q)^{*} = \left( \begin{array}{cccc} 3 & 1 & 1 & 0 \\ 1 & 3 & 0 & 1 \\ 1 & 0 & 3 & 1 \\ 0 & 1 & 1 & 3 \end{array} \right)$. 
\end{itemize} 
In this case we have three different decision criteria, one stable and two unstable.\\

For $n=3,4$ and $5$ we started with some simulations using the software SAGE \cite{SAGE}. We considered a set $A$ of $28900= 170^2$ points uniformly distributed on $\mathcal{T}$ and calculated the ordered form  $M_{n}(p,q)^{*}$ of each $(p,q)\in A$.\\

For $n=3$ we observe two criteria, $\mathscr{B}$ and $\mathscr{R}$, and by coloring the points in these regions by blue and red respectively we obtain the picture showing in Figure \ref{TrianguloN3}. 
\begin{figure}[h!]
\begin{center}
\includegraphics[scale=0.35]{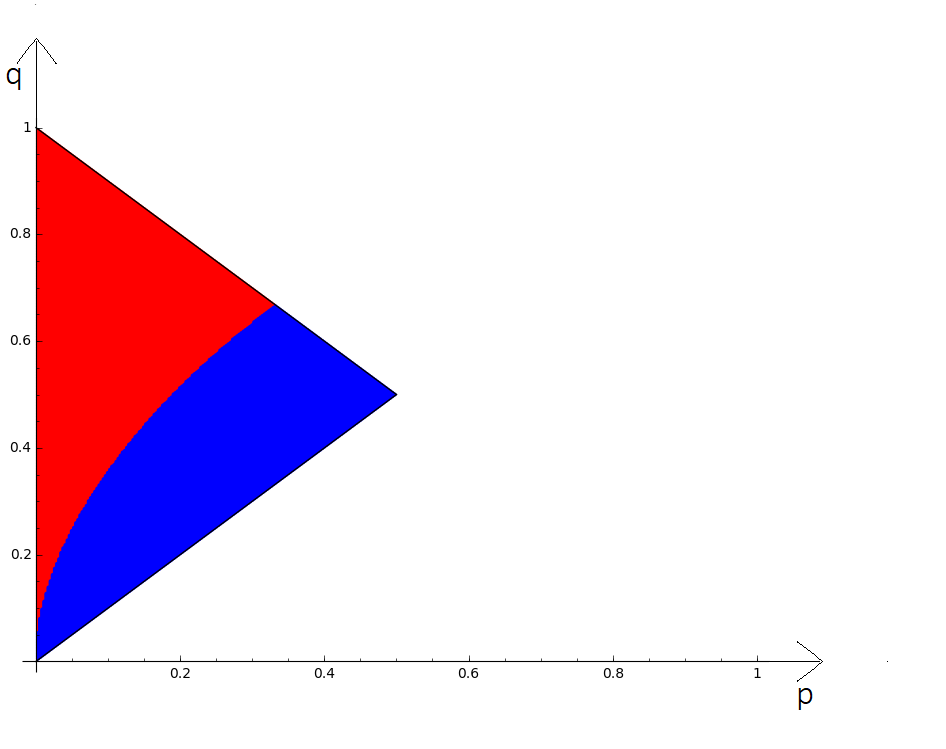}
\end{center}
\caption{Each color corresponds to a $3$-stable decision criterion for the BACs.}
\label{TrianguloN3}
\end{figure}
As we will prove in Theorem \ref{MainTheorem}, there are two stable criteria $\mathscr{B}$ and $\mathscr{R}$ (the connected components of the stable region $\mathcal{R}_{3}^{st}$) and there are three unstable decision criteria corresponding to the curves $p=0$ (the $Z$-channel), $p=q$ (the BSC) and the curve that separates the two connected components of the stable region. By considering the expression for the regions given in Theorem \ref{MainTheorem} we can find $(p,q)$ for some BACs which are representatives of these five criteria in the above order, let's say  $\left(\frac{1}{7},\frac{2}{7}\right)$, $\left(\frac{1}{7},\frac{4}{7}\right),  \left(0,\frac{1}{7}\right), \left(\frac{1}{7},\frac{1}{7}\right)$ and $\left(\frac{1}{7},\frac{3}{7}\right)$. We can then assert that any BAC is $3$-equivalent to one of these five channels.\\

%

We proceed similarly with the cases $n=4$ and $n=5$, observing three decision criteria for $n=4$ (Figure \ref{TrianguloN4}) and five decision criteria for $n=5$ (Figure \ref{TrianguloN5}). They correspond to the stable decision criteria and the curves separating these regions correspond to the four and six different unstable decision criteria for $n=4$ and $n=5$ respectively (see Theorem \ref{MainTheorem} in the next section). 

\begin{figure}[h!]
\begin{center}
\includegraphics[scale=0.35]{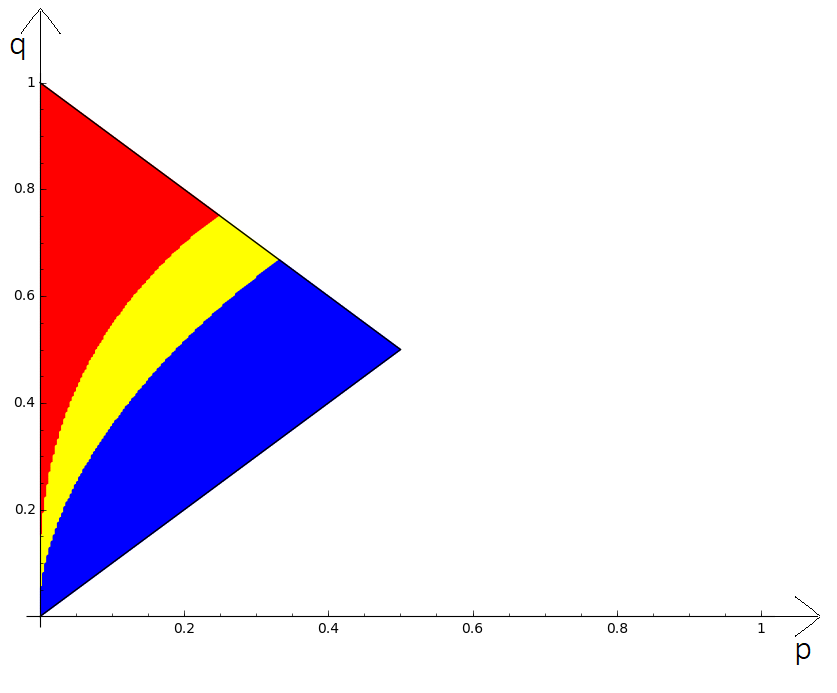}
\end{center}
\caption{Each color corresponds to a $4$-stable decision criterion for the BACs.}
\label{TrianguloN4}
\end{figure}

\begin{figure}[h]
\begin{center}
\includegraphics[scale=0.35]{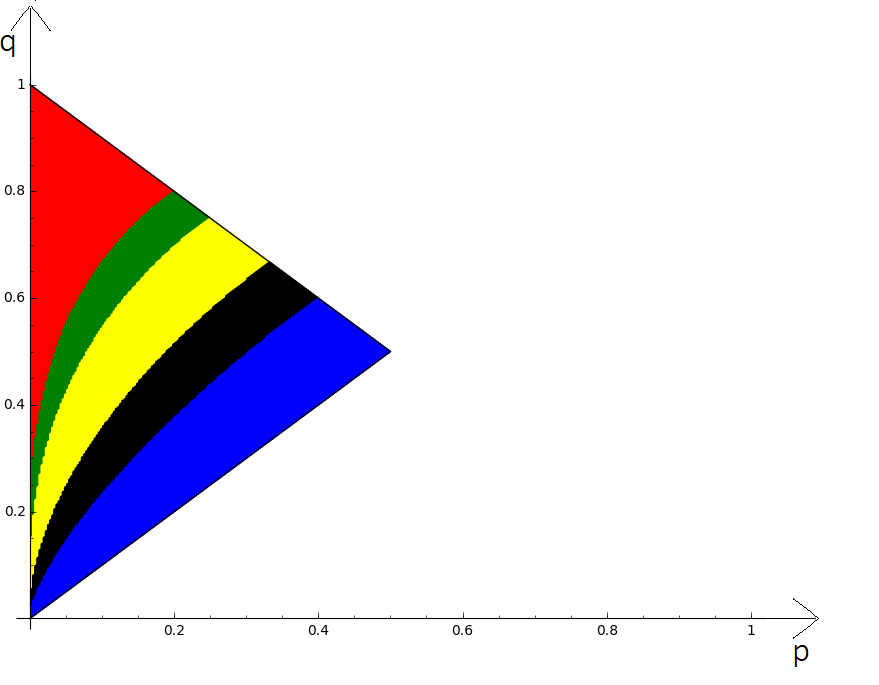}
\end{center}
\caption{Each color corresponds to a $5$-stable decision criterion for the BACs.}
\label{TrianguloN5}
\end{figure}

\section{Determining the $n$-decision criteria for the BACs}\label{SectionDecisionCriteria}

We start introducing a function which plays a fundamental role in describing the regions which determine the decision criteria for the BACs. This function also induces a natural distance between binary asymmetric channels as it is seen in Section \ref{SectionFurtherRemarks}.

\begin{definition}\label{DefBACfunction}
Let $S: \mathcal{T} \rightarrow [0,1]$ be the function given by $$ S(p,q)= \frac{\ln(1-p)-\ln(q)}{\ln(1-q)-\ln(p)},$$ if $p\neq 0$, and $S(p,q)=0$ if $p=0$. We refer to this function as the BAC-function.
\end{definition}

It is easy to check that in fact the image of $S$ is contained in the interval $[0,1]$ where the values $0$ and $1$ are attained by the extremes cases $p=0$ and $p=q$, respectively. Besides, this function is continuous in the connected set $\mathcal{T}$ and therefore we have $S(\mathcal{T})=[0,1]$.\\

Let $a$ and $b$ be integers with $0\leq a\leq b$. For $(p,q)\in \mathcal{T}$ we have 
\begin{equation}\label{LemmaEqForLevelCurves}
p^{a}(1-p)^{b}\geq q^{b}(1-q)^{a} \Leftrightarrow S(p,q)\geq \frac{a}{b},
\end{equation}
where equality holds in the left side if and only if it holds in the right side. The next lemma is a direct consequence of the above relation.

\begin{lemma}\label{CorImportante}
Let $a$ and $b$ be natural numbers with $a\leq b$, $a+b\leq n$ and $\eta=n-(a+b)\geq 0$. Consider the words $x=1^{a+\eta}0^{b},y=0^n, z=0^{\eta}1^{a+b} \in \mathbb{F}_{2}^{n}$. Then $S(p,q)\leq a/b$ if and only if $\mbox{Pr}(x|y)\leq \mbox{Pr}(x|z)$, where equality corresponds to equality. In particular, if $S(p_0,q_0)\leq a/b$ and $S(p_1,q_1)> a/b$ the channels $BAC(p_0,q_0)$ and $BAC(p_1,q_1)$ are not $n$-equivalent.
\end{lemma}

When we consider $p$ and $q$ as variables, the entries of the $n$-transition matrix $M_{n}(p,q)$ are polynomials in the variables $p$ and $q$. In fact, the entry of $M_{n}(p,q)$ corresponding to the conditional probability $\mbox{Pr}(x|y)$ ($x,y \in \mathbb{F}_{2}^{n}$) is equal to the polynomial $f(p,q)=p^a(1-p)^{b}q^{c}(1-q)^{d}$ where $a,b,c$ and $d$ correspond to the number of indices $i$ for which $(x_i,y_i)$ is equal to $(1,0),(0,0),(0,1)$ and $(1,1)$, respectively. In particular we have $a+b+c+d=n$ and the Hamming weight of $x$ is equal to $a+d$, which we will refer also as the weight of $f$ and denote by $\omega(f)$. We remark that two polynomials in the same row of $M_{n}(p,q)$ have the same weight. By the previous consideration we define the following sets of bivariate polynomials:
\begin{itemize}
\item $\mathcal{P}^{n}=\{f(p,q)=p^a(1-p)^{b}q^{c}(1-q)^{d}: a,b,c,d \geq 0, a+b+c+d=n\}$.
\item $\mathcal{P}_{k}^{n}=\{f \in \mathcal{P}^{n}: \omega(f)=k\}$ for $0\leq k \leq n$.
\end{itemize}
To determine the $n$-decision criterion corresponding to a BAC we only need to do comparisons between values in the same row of its $n$-transition matrix, this means comparisons of values of polynomials belonging to $\mathcal{P}_{k}^{n}$ for some $k$, $0\leq k \leq n$. The next lemma describes the stable region in terms of these sets.

\begin{lemma}\label{LemmaRegStable}
The $n$-stable region $\mathcal{R}_{n}^{st}$ is given by $$\mathcal{R}_{n}^{st}= \bigcap_{k=0}^{n}\ \ \bigcap_{\substack{f,g \in \mathcal{P}_{k}^{n}\\ f\neq g}}\{(p,q)\in \mathcal{T}: f(p,q)\neq g(p,q)\}.$$
\end{lemma}

\begin{proof}
We denote by $\widehat{\mathcal{R}}_{n}^{st}$ the set on the right side of the above equality. This set is open in $\mathcal{T}$ (since it is a finite intersection of open sets in $\mathcal{T}$), therefore if $(p_0,q_0)\in \widehat{\mathcal{R}}_{n}^{st}$ there is a ball $B_0$ with center at this point such that $B_0 \cap \mathcal{T}\subseteq \widehat{\mathcal{R}}_{n}^{st}$. Since $B_0\cap \mathcal{T}$ is connected, the signal of $f(p,q)-g(p,q)$ does not depends on $(p,q)\in B_0\cap \mathcal{T}$ (whenever $f,g \in \mathcal{P}_{k}^{n}$ for some $k$) and the same occurs with their associated decision criteria, so $\widehat{\mathcal{R}}_{n}^{st}\subseteq \mathcal{R}_{n}^{st}$. To prove the other inclusion we suppose by contradiction that there exists $(p,q) \not\in \widehat{\mathcal{R}}_{n}^{st} $ verifying $(p,q)\in \mathcal{R}_{n}^{st}$. Then, there exist two distinct polynomials $f,g \in \mathcal{P}_{k}^{n}$ for some $k: 0\leq k \leq n$ and a ball $B_0$ centered at $(p_0,q_0)$ such that $f(p_0,q_0)=g(p_0,q_0)$ and every point in $B_0\cap \mathcal{T}$ determines the same decision criterion, in particular $f(p,q)=g(p,q)$ for all $(p,q)\in B_0 \cap \mathcal{T}$. Since $B_0 \cap \mathcal{T}$ has interior points and two polynomials that coincide in an open set must be equal, we have $f=g$ which is a contradiction. 
\end{proof}

As we will prove next, the main property of the BAC-function from the point of view of this work, is that the curves which separate the regions corresponding to the stable and unstable criteria are level curves of this function associated with rational values. First, we prove some lemmas.

\begin{lemma}\label{LemmaRegUnstable}
Let $(p,q)\in \mathcal{T}$ be an $n$-unstable point for the BACs. Then $S(p,q)=\frac{a}{b}$ where $a,b$ are integers verifying $a\geq 0, b\geq 1,a\leq b, \gcd(a,b)=1$ and $a+b\leq n$.
\end{lemma}

\begin{proof}
By Lemma \ref{LemmaRegStable}, if $(p_0,q_0)\in \mathcal{T}$ is an $n$-unstable point for the BACs then there exists distinct polynomials $f_1,f_2 \in \mathcal{P}_{k}^{n}$ for some $k:0\leq k \leq n$ such that $f_1(p_0,q_0)=f_2(p_0,q_0)$. We write $f_i(p,q)=p^{a_i}(1-p)^{b_i}q^{c_i}(1-q)^{d_i}$ with $a_i+b_i+c_i+d_i=n$ and $a_i+d_i=k$ for $i=1,2$. Without loss of generality we suppose $a_1\geq a_2$. Let $a:=a_1-a_2=d_2-d_1$ and $b:=b_1-b_2=c_2-c_1$, then we have $$\frac{f_1(p,q)}{f_2(p,q)} = \left(\frac{p}{1-q} \right)^{a}\left(\frac{1-p}{q} \right)^{b},$$ or, equivalently, \begin{equation}\label{EqLevelCurves}
p^{a}(1-p)^{b}= \left(\frac{f_1(p,q)}{f_2(p,q)} \right)\cdot q^{b}(1-q)^{a}.
\end{equation}
Evaluating the Equation \ref{EqLevelCurves} for $(p,q)=(p_0,q_0)$ and using the relation (\ref{LemmaEqForLevelCurves}) we have $S(p_0,q_0)=\frac{a}{b}$. By our assumption we have $a\geq 0$, since $S(\mathcal{T})=[0,1]$ then $b\geq 1$ and $a+b=a_1-a_2+b_1-b_2\leq a_1+b_1\leq n$; simplifying common factors if necessary we can assume $\gcd(a,b)=1$. 
\end{proof}

Based on the above result, we introduce the following definition.

\begin{definition}
The weight of a non-negative rational $r$ (denoted by $\omega(r)$), is the sum of its numerator and denominator in the reduced expression of $r$. An $n$-critical value for the BAC-function $S$ (where $n\geq 2$) is a rational number $r\in [0,1]$ with $\omega(r)\leq n$. The set of all $n$-critical values for $S$ is denoted by $\mathcal{D}_n$. 
\end{definition}

\begin{corollary}
If we write the set of the $n$-critical values for $S$ as $\mathcal{D}_{n}=\{r_0=0<r_1<\cdots<r_t=1\}$ and denote by $R_n(r_{i})=\{(p,q) \in \mathcal{T}: r_{i}<S(p,q)<r_{i+1}\}$ then $R_n(r_{i})\subseteq \mathcal{R}_n^{st}$ for $0\leq i < t$.
\end{corollary}

\begin{lemma}\label{LemmaContinuity}
Let $(p,q)\in \mathcal{T}$ and $M_{n}(p,q)$ be the $n$-transition matrix for the channel $\mbox{BAC}(p,q)$ (seeing as an element of $\mathbb{R}^{n^2}$). The function $\phi: \mathcal{R}_{n}^{st}\rightarrow \mathbb{R}^{n^2}$ given by $\phi(p,q)=M_{n}(p,q)^{*}$ is continuous.
\end{lemma}

\begin{proof}
Let $(p_0,q_0)\in \mathcal{R}_{n}^{st}$ and $f,g \in \mathcal{P}_{k}^{n}$ for some $k$, $0\leq k \leq n$ with $f\neq g$. By Lemma \ref{LemmaRegStable} we have $f(p_0,q_0)\neq g(p_0,q_0)$, then there exists $\epsilon=\epsilon(f,g)>0$ such that the sign of $f(p,q)-g(p,q)$ does not depend on $(p,q)\in B\left((p_0,q_0),\epsilon \right)\cap \mathcal{T}$. If $\epsilon = \min\{\epsilon(f,g): f,g \in \mathcal{P}_{k}^{n}, f\neq g, 0\leq k \leq n\}$ and denoting by $B$ the $\epsilon$-ball centered at $(p_0,q_0)$ we have that $f(p,q)>g(p,q) \Leftrightarrow f(p_0,q_0)>g(p_0,q_0)$ for all $(p,q) \in B \cap \mathcal{T}$ and for all $f,g \in \mathcal{P}_{k}^{n}, f\neq g, 0\leq k \leq n$. Therefore $M_{n}(p,q)^{*}= M_{n}(p_0,q_0)^{*}$, so $\phi$ is locally constant and in particular continuous. 
\end{proof}

\begin{lemma}\label{LemmaConexity}
Let $S$ be the BAC-function. Then $S^{-1}(I)$ is a connected set for every interval $I \subseteq [0,1]$.
\end{lemma}

\begin{proof}
Let $\tau \in (0,1)$ and $g_\tau: [0,\frac{\tau}{2}] \rightarrow [0,1]$ be the  function given by $g_{\tau}(p)=S(p,\tau-p)$. We affirm that it is increasing (in the variable $p$). To prove this we consider $s\in [0,1]$ and the function $f_{s}(p)=p^s(1-p)-q(1-q)^s$ (where $q=\tau-p$). Since $1-q>p$ we have $p^{s-1}>(1-q)^{s-1}$ and $(1-q)^{s}>p^{s}$, therefore
$$f_{s}'(p)=s\left(p^{s-1}(1-p)-q(1-q)^{s-1}\right)+(1-q)^{s}-p^{s}$$ $$ > 
s(1-q)^{s-1}(1-p-q)+(1-q)^{s}-p^{s}>0.$$
Since $f_{s}(0)=-q(1-q)^{s}<0$ and $f_{s}(\tau/2)=\tau/2 \left(1-\tau/2 \right)\left( (\tau/2)^{s-1}-(1-\tau/2)^{s-1} \right)>0$ (because $1-\tau/2 > \tau/2$ and $s-1<0$), then there is an unique $p\in (0,\tau/2)$ such that $f_{s}(p)=0$, or equivalently, such that $g_{\tau}(p)=S(p,\tau-p)=s$. Therefore $g_\tau: [0,\frac{\tau}{2}] \rightarrow [0,1]$ is increasing since it is a continuous bijection with $g_\tau(0)=0$ and $g_{\tau}(\tau/2)=1$. Let $I\subseteq [0,1]$ be an interval and $(p_0,q_0),(p_1,q_1) \in \mathcal{T}$ be two points in $S^{-1}(I)$. We denote by $s_i=S(p_i,q_i)$ and $\tau_i=p_i+q_i$ for $i=0,1$. Without loss of generality we assume $s_0\leq s_1$. Since $g_{\tau_{0}}(p_0)=S(p_0,q_0)=s_0<s_1$ there exists $t_0>0$ such that $g_{\tau_{0}}(p_0+t_0)=s_1$ and $g_{\tau_{0}}(p_0+t)\in (s_0,s_1)$ for all $t\in (0,t_0)$ (in particular $(p_0+t,q_0-t)\in S^{-1}(I)$ for all $t\in [0,t_0]$). Taking the path $\alpha:[0,t_0]\rightarrow \mathcal{T}$ given by $\alpha(t)=(p_0+t,q_0-t)$ and $\beta$ the segment of curve in $S^{-1}(s_1)$ from $(p_0+t_0,q_0-t_0)$ to $(p_1,q_1)$, we have that the concatenation path $\beta * \alpha$ is a path connecting $(p_0,q_0)$ with $(p_1,q_1)$. Therefore $S^{-1}(I)$ is path-connected and then connected.
\end{proof}

\begin{lemma}\label{Lemma6}
Let $(p_0,q_0)$ and $(p_1,q_1)$ be two points in $\mathcal{T}$.
If $S(p_0,q_0)=S(p_1,q_1)$ then $BAC(p_0,q_0)\eqv BAC(p_1,q_1)$ for all $n \geq 1$.
\end{lemma}

\begin{proof}
We suppose by contradiction that $(p_0,q_0)$ and $(p_1,q_1)$ are not $n$-equivalent, so there exists two polynomials $f,g \in \mathcal{P}_{k}^{n}$ for some $k$ verifying $f(p_0,q_0)<g(p_0,q_0)$ and $f(p_1,q_1)\geq g(p_1,q_1)$. As in the proof of Lemma \ref{LemmaRegUnstable}, the polynomials $f$ and $g$ must verify an equation similar to Equation (\ref{EqLevelCurves}): 
\begin{equation}\label{EqLevelCurves2}
p^{a}(1-p)^{b}= \left(\frac{f(p,q)}{g(p,q)} \right)\cdot q^{b}(1-q)^{a}
\end{equation}
 for some integers $a\geq0$ and $b\geq 1$ satisfying $a\leq b$ and $a+b\leq n$. We consider a path $\alpha$ contained in the curve $S^{-1}(r)$ connecting $(p_0,q_0)$ with $(p_1,q_1)$, since $f/g<1$ in $(p_0,q_0)$ and $f/g\geq 1$ in $(p_1,q_1)$, by continuity there exists an intermediate point $(p_2,q_2)$ in $\alpha$ for which $f/g=1$. Evaluating Equation (\ref{EqLevelCurves2}) in $(p,q)=(p_2,q_2)$ and using the relation (\ref{LemmaEqForLevelCurves}) we have $S(p_2,q_2)=a/b$. Since $(p_2,q_2)$ belongs to $\alpha$ which is contained in the level curve $S^{-1}(r)$ we have $S(p_2,q_2)=r$, then $r=a/b$. Substituting $(p,q)=(p_0,q_0)$ in Equation (\ref{EqLevelCurves2}) and using the relation \ref{LemmaEqForLevelCurves} we have $r=S(p_0,q_0)<a/b$ which is a contradiction. Therefore $(p_0,q_0)$ and $(p_1,q_1)$ must be $n$-equivalent.
\end{proof}

Now we are ready to state the main result of this paper.

\begin{theorem}\label{MainTheorem}
Let $S:\mathcal{T} \rightarrow [0,1]$ be the BAC-function $$S(p,q)=\frac{\ln(1-p)-\ln(q)}{\ln(1-q)-\ln(p)},$$ for $p\neq 0$, $S(0,q)=0$ and $\mathcal{D}_n=\{0=r_0<r_1<\cdots < r_{t_n}=1\}$ its set of $n$-critical values ($n\geq 2$). We consider the level curves $\gamma_i = S^{-1}(r_i)$ for $0\leq i \leq t_n$ and the regions $R_n(r_{i})=\{(p,q) \in \mathcal{T}: r_{i}<S(p,q)<r_{i+1}\}$ for $0\leq i < t_n$. 
Then 
\begin{equation}\label{Equationtn}
t_n=1+\frac{1}{2}\sum_{k=3}^{n}\varphi(k)
\end{equation} where $\varphi$ denotes the Euler's totient function and there are exactly $t_n$ stable decision criteria of order $n$ for the BACs, which are given by $\{R_n(r_{i}):0\leq i < t_n \}$ and exactly $t_n+1$ unstable decision criteria of order $n$ for the BACs given by $\{\gamma_i : 0\leq i \leq t_n\}$.
\end{theorem}

\begin{proof}
Consider $\mathcal{T}$ written as the disjoint union 
$$\mathcal{T}= \biguplus_{i=0}^{t-1}R_n(r_i) \uplus \biguplus_{i=0}^{t}\gamma_i.$$ We have to prove that each $R_n(r_i)$ and each $\gamma_i$ is a $n$-decision criterion (i.e. an equivalent class in $\Delta_n$). If $(p_0,q_0)\in \gamma_i$ for some $i: 0\leq i \leq t$, by Lemma \ref{Lemma6} and Lemma \ref{CorImportante} we have $BAC(p_1,q_1)\eqv BAC(p_0,q_0)$ if and only if $(p_1,q_1)\in \gamma_i$, then each $\gamma_i$ is a decision criterion. We consider now a point $(p_0,q_0)\in R_n(r_i)$ for some $i: 0\leq i <t$ and $(p_1,q_1)\in \mathcal{T}\setminus R_n(r_i)$. If $(p_1,q_1)\in \gamma_j$ for some $j$, since each $\gamma_j$ is a decision criterion we have that $BAC(p_1,q_1)$ and $BAC^{n}(p_0,q_0)$ are not $n$-equivalent. Otherwise $(p_1,q_1)\in R_n(r_j)$ for some $j: 0\leq j < n, j\neq i$ and there exists $r_k \in \mathcal{D}_n$ such that $S(p_0,q_0)<r_k<S(p_1,q_1)$ or $S(p_1,q_1)<r_k<S(p_0,q_0)$. In both cases, by Lemma \ref{CorImportante} the channels $BAC(p_1,q_1)$ and $BAC^{n}(p_0,q_0)$ are not $n$-equivalent. It only remains to prove that if $(p_1,q_1)\in R_n(r_i)$ the channels $BAC(p_1,q_1)$ and $BAC(p_0,q_0)$ are $n$-equivalent. We consider the function $\phi: \mathcal{R}_{n}^{st}\rightarrow \mathbb{R}^{n^2}$ given by $\phi(p,q)=M_{n}(p,q)^{*}$ where $M_{n}(p,q)$ denotes the $n$-transition matrix for the channel $\mbox{BAC}(p,q)$. By Lemma \ref{LemmaConexity}, $R_n(r_i)$ is a connected set (since $R_n(r_i)=S^{-1}(I)$ for $I=(r_i,r_{i+1})$), and by Lemma \ref{LemmaRegUnstable}, $R_n(r_i)$ is contained in the stable region $\mathcal{R}_{n}^{st}$. By Lemma \ref{LemmaContinuity}, the set $\phi(R_n(r_i))\subseteq \mathcal{M}_{n}(\mathbb{Z})$ is connected and since $\mathcal{M}_{n}(\mathbb{Z})$ is discrete, there exists $M \in \mathcal{M}_{n}(\mathbb{Z})$ such that $\phi(R_n(r_i))=\{M\}$. Therefore $M_n(p_1,q_1)^{*}=M_n(p_0,q_0)^{*}=M$ and $BAC^{n}(p_1,q_1)\sim BAC^{n}(p_0,q_0)$.\\

Since the sets $R_n(r_i)$ are open for $0\leq i <t_n$ and the sets $\gamma_i$ have empty interior they correspond to the stable and unstable criteria respectively. To derive the formula for $t_n$ we consider the decomposition into disjoint sets: $\mathcal{D}_n= \biguplus_{k=1}^{n}\mathcal{D}_{k}^{o}$ where $\mathcal{D}_{k}^{o}=\{a/b\in \mathbb{Q}^{+}: a\leq b, \gcd(a,b)=1, a+b= k\}$. We have $\# \mathcal{D}_{k}^{o}=1$ for $k=1,2$. For $k\geq 3$ if $a/b \in \mathcal{D}_{k}^{o}$ then $a<b$ and in this case: $$\# \mathcal{D}_{k}^{o}=\frac{1}{2}\cdot \# \left\{(a,b)\in \mathbb{N}^{2}: \gcd(a,b)=1, a+b=k\right\}$$ $$\quad\quad\ \ = \frac{1}{2}\cdot\# \left\{(a,b)\in \mathbb{N}^{2}: \gcd(a,k)=1, a+b=k\right\}$$ $$\qquad \qquad= \frac{1}{2}\cdot \# \left\{a\in \mathbb{N}: \gcd(a,k)=1, a\leq k\right\}=\frac{1}{2}\cdot \varphi(k).$$ Then, for $n\geq 3$ we have: 
\begin{align*}
t_n &=\# \mathcal{D}_{n}-1 = \sum_{k=1}^{n}\# \mathcal{D}_{k}^{o}-1\\ &=2+\frac{1}{2}\cdot \sum_{k=3}^{n}\varphi(k)-1 = 1+ \frac{1}{2}\cdot \sum_{k=3}^{n}\varphi(k).
\end{align*}
For $n=2$ we have $\#\mathcal{D}_{2}=\#\left\{(0,1),(1,1)\right\}=2$ and the above formula also holds in this case. 
\end{proof}

\begin{corollary}
Let $(p_0,q_0),(p_1,q_1) \in \mathcal{T}$. The channels $BAC(p_0,q_0)$ and $BAC(p_1,q_1)$ are $\infty$-equivalent (i.e. $n$-equivalent for all $n$) if and only if $S(p_0,q_0)=S(p_1,q_1)$.
\end{corollary}

\begin{proof}
If $S(p_0,q_0)<S(p_1,q_1)$ there exists a rational number $r \in \mathcal{D}_n$ for some $n\geq 1$ large enough such that $S(p_0,q_0)<r<S(p_1,q_1)$, then the channels $BAC(p_0,q_0)$ and $BAC(p_1,q_1)$ are not $n$-equivalent. The converse is consequence of Lemma \ref{Lemma6}.
\end{proof}

\begin{corollary}
A point $(p,q)\in \mathcal{T}$ is a stable point of order $n$ for all $n\geq 1$ if and only if $S(p,q)$ is an irrational number.  
\end{corollary}

Using the average order formula for the Euler's totient function $\varphi$ (see for example Theorem 3.7 of \cite{Apostol76}) we obtain the following corollary.

\begin{corollary}
The number of stable decision criteria of order $n$ for the BACs grows quadratically with $n$. More explicitly, it is given by $\frac{3}{\pi^2}\cdot n^2 + O\left(n\cdot \ln n \right)$.
\end{corollary}

\begin{example} For $n=5$ (see Figure \ref{TrianguloN5}), we have the set of critical values $\mathcal{D}_{5}=\{ 0,\ 1/4,\ 1/3,\ 1/2,\ 2/3,\ 1\}$ which correspond to the level curves of the BAC-function describing the unstable sets. Those curves can be seen from left to right according to the order in $\mathcal{D}_{5}$. The five stable regions are the ones bounded by these curves.
\end{example}

\begin{example} For $n=9$, we have $t_9= 29$ decision regions, $15$ instable, associated to the critical set $\mathcal{D}_{9}=\{0,\ 1/8,\ 1/7,\ 1/6,\ 1/5,\ 1/4,\ 2/7,\ 1/3,\ 2/5,\ 1/2,\ 3/5,\ 2/3,\ 3/4,$ $ 4/5,\ 1\}$ and $14$ stable, situated between the level curves of the BAC-function attached to values in $\mathcal{D}_{9}$. 
\end{example}

To conclude this section we present some situations of decoding problems that can be solved using Theorem \ref{MainTheorem}.

\begin{example}\label{ExampleMismatching1} Let $W$ be a memoryless binary asymmetric channel. Suppose that a series of measurements was performed in order to obtain the transition probabilities $(p,q)$ for $W$ and the following values were obtained: $p=0.212$ and $q=0.531$ with a possible error of at most $\varepsilon=0.001$ in both probabilities. We want to determine the maximum possible length $N$ such that there is no risk of mismatching when we perform MLD on $W$ with respect to $(p,q)=(0.212,0.531)$ and binary codes with block length $n\leq N$. In term of $n$-equivalence of channels this means to find the maximum value of $N$ such that the channels $W$ and $BAC(0.212,0.531)$ are $n$-equivalent for all $n\leq N$. By our assumption, the real transition probabilities $(p_0,q_0)$ for $W$ (which is unknown) belongs to the square $I=[0.211,0.213] \times [0.530,0.532]$. By Theorem \ref{MainTheorem}, the problem is reduced to find the maximum value of $N$ such that the critical set $\mathcal{D}_{N}=\{\frac{a}{b}: 0 \leq b \leq a, a+b\leq N\}$ has empty intersection with the interval $S(I)=[S(0.211,0.532), S(0.213,0.530)]=[0.4947499..,0.4995337]$. The maximum value of $N$ is $144$ since $47/95 = 0.4947368..$ and $1/2=0.5$ are consecutive elements in $\mathcal{D}_{144}$ but $0.4947499<48/97 <0.4995337$ with $48+97=145$. Thus, if we implement MLD with respect to the measured approximated value $(p,q)=(0.212,0.531)$, restricted to codes with block length $n\leq 144$, there is no risk of mismatched decoding since in this case $BAC(p,q)\eqv BAC(p_0,q_0)$.
\end{example}

\begin{example}\label{ExampleMismatching2} Let $W=BAC(p_0, q_0)$ with $p_0=0.314$, $q_0=0.594$ and suppose that only binary codes with block length at most $n=32$ are considered. We call $\varepsilon>0$ an admissible error if $BAC(p,q)\eqn{20}W$ for all $(p,q)\in\mathcal{T}$ such that $|p-p_0|<\varepsilon$ and $|q-q_0|<\varepsilon$. We want to find the greatest admissible error. For each $\varepsilon>0$ we consider the box $I_\varepsilon = [p_0-\varepsilon,p_0+\varepsilon]\times [q_0-\varepsilon, q_0+\varepsilon]$. By Theorem \ref{MainTheorem}, the problem is equivalent to find the greatest $\varepsilon>0$ such that $S(I_\varepsilon) \cap \mathcal{D}_{32}= \emptyset$. Since $5/9$ and $9/16$ are consecutive elements of $\mathcal{D}_{32}$ satisfying $5/9 < S(p_0,q_0)=0.56039...<9/16$, and the endpoints of $S(I_\varepsilon)$ are $S(p_0-\varepsilon, q_0+\varepsilon)$ and $S(p_0+\varepsilon, q_0-\varepsilon)$ , it suffices to find $\varepsilon_0,\varepsilon_1 >0$ such that $S(p_0-\varepsilon_0, q_0+\varepsilon_0)=5/9$ and $S(p_0+\varepsilon_1, q_0-\varepsilon_1)=9/16$, and take the minimum of these two values. By direct calculation we obtain $\varepsilon_0=0.0019779637..$ and $\varepsilon_1= 0.0008585765..$, thus $\varepsilon=\min\{\varepsilon_0,\varepsilon_1\}= 0.0008585765..$.
\end{example}

\section{Further remarks}\label{SectionFurtherRemarks}

\subsection{On the BAC-function and the parameter space for the BACs}

To study the different $n$-decision criteria for the BACs, we choose the parameter space $\mathcal{T}=\{(p,q)\in [0,1]: p+q<1, 0\leq p \leq q\}\setminus \{(0,0)\}$ and use the BAC-function to describe the regions determined by these criteria. In the first part of this section we discuss what happens when we remove the restriction $p+q<1$ and $0\leq p \leq q$ and what is the role of the BAC-function in these cases. Next, we show how to obtain a natural distance between BACs in such a way that the BAC-function measures how far a channel is from the binary symmetric channel, in this sense the BAC-function provides a measure of the asymmetry of the channel.\\

Consider a BAC with transition probabilities $p,q \in [0,1]^2$. A BAC is \emph{reasonable} (in the sense of \cite{FW16}) if their transition probabilities verify $\Pr(0|0)>\Pr(0|1)$ and $\Pr(1|1)>\Pr(1|0)$. Note that this is equivalent to the condition $p+q<1$. It is not difficult to check, using Equation (\ref{EqLevelCurves})), the following facts:
\begin{itemize}
\item $p+q<1 \Leftrightarrow \Pr(x|x)>\Pr(x|y)$,  $\ \forall x,y \in \F_2^n, x\neq y$;
\item $p+q>1 \Leftrightarrow \Pr(x|x)<\Pr(x|y)$, $\ \forall x,y \in \F_2^n, x\neq y$;
\item $p+q=1 \Leftrightarrow \Pr(x|x)=\Pr(x|y)$, $\ \forall x,y \in \F_2^n, x\neq y$.
\end{itemize}

As a consequence, we have that if two channels $\mbox{BAC}(p,q)$ and $\mbox{BAC}(p',q')$ are $n$-equivalent. then the sign of $1-p-q$ and $1-p'-q'$ is the same. The case $p+q=1$ corresponds to the completely noisy channel, which are not interesting from the coding/decoding point of view, since $\Pr(x|y)=\Pr(x|z)$ for all $x,y,z \in \F_2^n$. These channels are $n$-equivalent among them, for all $n\geq 1$. The case $p+q<1$ can be decomposed into two regions $\mathcal{T}$ and $\mathcal{T}'=\{(p,q)\in [0,1]^{2}: p+q<1,\ 0\leq q \leq p\}\setminus \{(0,0)\}$ which are symmetric one to the other (via the map $(p,q)\mapsto (q,p)$).

We observe that the BAC-function $S$ can be extended to a function $\widehat{S}:\mathcal{T}\cup \mathcal{T}'\rightarrow [0,+\infty]$ defining $\widehat{S}(q,p)=1/S(p,q)$ for $(q,p) \in \mathcal{T}'$ if $pq\neq 0$ and $\widehat{S}(p,0)= + \infty$. This extension is continuous and verifies $\widehat{S}(\mathcal{T}')=[1,+\infty]$. By relation (\ref{LemmaEqForLevelCurves}), $S(p,q)\leq 1$ if and only if $p(1-p)\leq q(1-q)$, if and only if $p^{n-1}(1-p)\leq p^{n-2}q(1-q)$. Therefore, for $x=1^{n-1}0$, $y=0^n$ and $z=0^{n-2}1^2$ we have:

\begin{itemize}
\item $\Pr(x|y)< \Pr(x|z)$ if $(p,q)\in \mathcal{T}$ with $p \neq q$,
\item $\Pr(x|y)> \Pr(x|z)$ if $(p,q)\in \mathcal{T}'$ with $p\neq q$ and
\item $\Pr(x|y)> \Pr(x|z)$ if $p=q$.
\end{itemize}
Thus, we conclude that the triangles $\mathcal{T}$ and $\mathcal{T}'$ have no common criteria decision except for those points corresponding to the BSC. Since the $i$-th row of the $n$-transition matrix $M_n(p,q)$ is just the $(2^n+1-i)$-th row of $M_n(q,p)$ in reverse order, then $\mbox{BAC}(p,q)\eqv \mbox{BAC}(p',q')$ if and only if $\mbox{BAC}(q,p)\eqv \mbox{BAC}(q',p')$ for all $(p,q),(p',q')\in \mathcal{T}$. By the above consideration we have the following proposition.

\begin{proposition}
Let $S:\mathcal{T}\cup \mathcal{T}' \rightarrow [0,+\infty]$ be the BAC-function defined as above and $\mathcal{D}_n=\{0=r_0<r_1<\cdots < r_{t_n}=1\}$ its set of $n$-critical values ($n\geq 2$) where $t_n=1+\frac{1}{2}\sum_{k=3}^{n}\varphi(k)$. There are exactly $2t_n$ stable decision criteria of order $n$ for the $\mbox{BAC}(p,q)$ with $(p,q)\in \mathcal{T}\cup \mathcal{T}'$ and exactly $2t_n-1$ unstable decision criteria of order $n$. The $n$-stable criteria are given by $R_n(r_{i})=\{(p,q) \in \mathcal{T}: r_{i}<S(p,q)<r_{i+1}\}$ and $R_n(r_{i}^{-1})=\{(p,q) \in \mathcal{T}: r_{i+1}^{-1}<S(p,q)<r_{i}^{-1}\}$ for $0\leq i < t_n$. The $n$-unstable criteria are given by the level curves $S^{-1}(r_i)$ and $S^{-1}(r_i^{-1})$ for $0\leq i \leq t_n$.
\end{proposition}

The function $S$ is constant when restricted to a criterion $\mathcal{A}\subseteq \mathcal{T}\cup \mathcal{T}'$. Thus it defines an injective function on the equivalent classes $S: \Delta_n \rightarrow [0,+\infty]$ such that $S(\mathcal{A}):=S(p,q)$ for any $(p,q)\in \mathcal{A}$. If $\Delta_n^{*}$ denotes the set of all $n$-decision criteria for the BACs except for those corresponding to the $Z$-channels (i.e. when $pq=0$), then the function $d: \Delta_n^{*} \times \Delta_n^{*} \rightarrow [0,+\infty)$ given by $d(\mathcal{A},\mathcal{B})=|\ln S(\mathcal{A}) - \ln S(\mathcal{B})|$ defines a metric in the $n$-decision criteria space for the BACs. In particular if $\mathcal{B}$ denotes the criterion corresponding to the BSC then $d(\mathcal{A},\mathcal{B})=|\ln S(\mathcal{A})|$ can be interpreted as a measure of how asymmetric a channel is.\\

Since the ordered form of the $n$-transition matrix associated with the completely noisy channels ($p+q=1$) is the null matrix (because in this case $\Pr(x|y)=p^{w_{H}(x)}(1-p)^{n-w_{H}(x)}$ depends only on the Hamming weight of $x$ and not on $y$) and this is the only situation for which it happens, then the points in the line $p+q=1$ correspond to a single criterion when considering BACs with $(p,q)\in [0,1]^{2}$. If $p+q<1$ and $p'+q'>1$, the channels $\mbox{BAC}(p,q)$ and $\mbox{BAC}(p',q')$ are not $n$-equivalent, since the first is reasonable and the last is not. 

Let $T^{-}=\mathcal{T}\cup \mathcal{T}'=\{(p,q)\in [0,1]^2: p+q<1\}\setminus \{(0,0)\}$ and $T^{+}=\{(p,q) \in [0,1]^2:p+q>1)\}\setminus\{(1,1)\}$. The involution $\phi(p,q)=(1-q,1-p)$ maps the triangle $T^{-}$ into $T^{+}$ and the curves $\gamma_{a/b}: p^{a}(1-p)^{b}=q^{b}(1-q)^{a}$ into itself. Moreover, the $i$-th rows of $M_{n}(p,q)$ and $M_{n}(1-q,1-p)$ are the same but in reverse order. Therefore for $(p,q),(p',q') \in T^{+}$, the channels $\mbox{BAC}(p,q)$ and $\mbox{BAC}(p',q')$ are $n$-equivalent if and only if $\mbox{BAC}(1-q,1-p)$ and $\mbox{BAC}(1-q',1-p')$ are $n$-equivalent. 

We conclude that if we consider BACs with parameters $(p,q)\in [0,1]^{2}$, the $n$-stable criteria are the regions bounded by the edges of the square $[0,1]^2$ and the curves $\gamma_{a/b}$ where $a/b$ is a positive rational number with $a,b \in \mathbb{Z}^+$, $\gcd(a,b)=1$ and $a+b\leq n$. We also remark that the level curves $\gamma_{a/b}$ contain the line $p+q-1=0$ but if we divide the equation $p^{a}(1-p)^{b}- q^{b}(1-q)^{a}=0$ which defines $\gamma_{a/b}$ (see Equation (\ref{LemmaEqForLevelCurves})) by $p+q-1$ we obtain irreducible curves $\widehat{\gamma}_{a/b}$ (see Figure \ref{Square}).

\begin{figure}[h!]
\begin{center}
\includegraphics[scale=0.3]{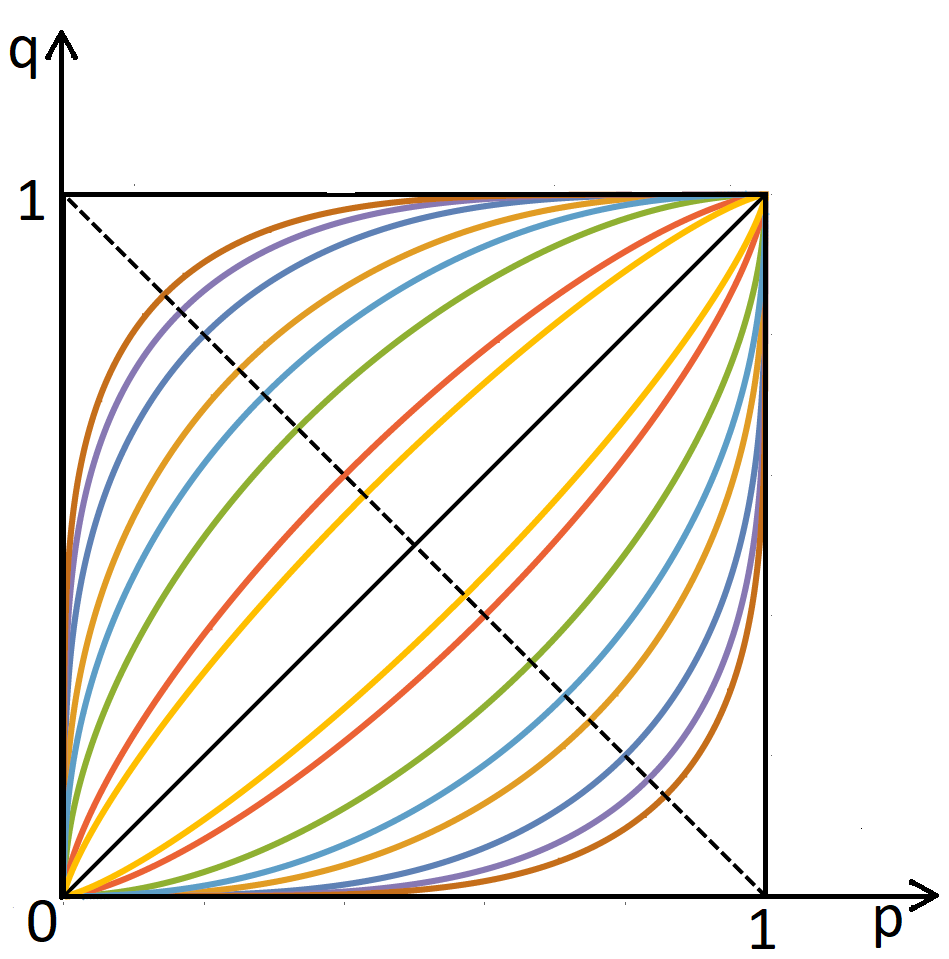}
\end{center}
\caption{Decision criteria (regions) of order $7$ for the BACs with $(p,q)$ in the unit square: the line $p+q=1$ (dotted) and the curves 
$\widehat{\gamma}_{q}$ and $\widehat{\gamma}_{q^{-1}}$ for $q=1/6$ (brown), 
$1/5$ (violet), $1/4$ (blue), $1/3$ (orange), $2/5$ (sky blue), $1/2$ (green), 
$2/3$ (red), $3/4$ (yellow) and $1$ (black). Curves corresponding to reciprocal values have the same color.}
\label{Square}
\end{figure}

\subsection{The most probable BACs are those next to the BSC}

By the previous discussion, without loss of generality we can restrict our parameter space to $\mathcal{T}=\{(p,q)\in [0,1]^2: p+q<1,\ p\leq q\}\setminus \{(0,0)\}$ . Let $\mathcal{A}\subseteq \mathcal{T}$ be a $n$-decision criterion for the BACs. If we choose a point $(p,q)\in \mathcal{T}$ uniformly at random and consider the channel $W=\mbox{BAC}(p,q)$, the probability $\Pr\left(W\eqv \mathcal{A}\right)=4\cdot a(\mathcal{A})$ where $a(\mathcal{A})$ is the area of the region corresponding to the criterion $\mathcal{A}$. In this sense, the area of a given $n$-decision criterion for the BACs is a measure of how probably this criterion is to be chosen (assuming uniform distribution on $(p,q)$). By Theorem \ref{MainTheorem}, the area of $\mathcal{A}=\{(p,q)\in \mathcal{T}: r_0< S(p,q) < r_1 \}$ is  $a(r_1)-a(r_0)$ where $r_1$ and $r_0$ are consecutive rational numbers in $\mathcal{D}_{n}$ and $A(r)=a(R_r)$ is the area of the region $R_r=\{(p,q)\in \mathcal{T}: 0< S(p,q) < r\}$. This area is equal to $$A(r)= \iint_{R_r} 1\ dp\ dq.$$ Let $r=a/b$ where $a$ and $b$ and coprime positive integers with $a<b$. Applying the change of variable formula for double integrals with $p=\frac{u-1}{uv-1}, q=\frac{v-1}{uv-1}$ and after some calculations we obtain:
\begin{equation}\label{AreaEquation}
A(r)= \int_{0}^{1} \frac{b(x^a-1)^2x^{b-1}}{2(x^{a+b}-1)^2} dx.
\end{equation}
Since $x=1$ is a zero of order $2$ of the numerator, the integral is a proper integral. In some cases, a primitive for the integrand can be calculated explicitly, for example when $r=1/2$ and $r=1/3$ obtaining $A(1/2)= \frac{1}{3}-\frac{\sqrt{3}\pi}{27}$ and $A(1/3)= \frac{3}{8}-\frac{3\pi}{32}$. In the other cases we have use the software Wolfram Mathematica \cite{WolframMath} to calculate the integral (\ref{AreaEquation}) numerically after some reductions.\\

The $n$-stable criterion for the BACs nearest the BSC is denoted by $\mathcal{A}_{QS}^{n}$. It is given by the criterion corresponding to the channels $\mbox{BAC}(p,q)$ satisfying $r_n:=\frac{2n-3-(-1)^{n}}{2n+1-(-1)^{n}}<S(p,q)<1$. We refer to these channels as \emph{$n$-quasi-symmetric channels}. Clearly $a(\mathcal{A}_{QS}^{n})\to 0$ when $n \to \infty$ (since $r_n \to 1$). In the next table we show the percentages (rounded to the nearest integer) represented by the different $n$-stable criteria for the BACs, for $3\leq n \leq 7$.

\begin{center}
\begin{tabular}{c|c|c|c|c|c|c|c|c|c|c}
\multicolumn{11}{l}{\hspace{10mm}$0$ \hspace{3.5mm}$\frac{1}{6}$\hspace{2.7mm}  $\frac{1}{5}$\hspace{3.7mm}$\frac{1}{4}$\hspace{6mm}$\frac{1}{3}$\hspace{3.3mm}$\frac{2}{5}$\hspace{6.2mm}$\frac{1}{2}$\hspace{5.3mm}$\frac{2}{3}$\hspace{4.1mm}$\frac{3}{4}$\hspace{5.5mm}$1$} \\ \cline{1-10}
$n=3$ & \multicolumn{6}{|c|}{$53$} & \multicolumn{3}{|c|}{$47$} & \\ \cline{1-10}
$n=4$ & \multicolumn{4}{|c|}{$32$} &\multicolumn{2}{|c|}{$21$} &\multicolumn{3}{|c|}{$47$} & \\ \cline{1-10}
$n=5$ & \multicolumn{3}{|c|}{$22$} & $11$ &\multicolumn{2}{|c|}{$21$} &$18$ &\multicolumn{2}{|c|}{$29$} & \\ \cline{1-10}
$n=6$ & \multicolumn{2}{|c|}{$16$} &$6$ &$11$ &\multicolumn{2}{|c|}{$21$} &$18$ &\multicolumn{2}{|c|}{$29$} & \\ \cline{1-10}
$n=7$ & $12$ & $4$ & $6$& $11$& $8$ & $12$& $18$& $8$ & $21$ &   \\ \cline{1-10}
\end{tabular}
\end{center}
\vspace{3mm}

For example, for $n=3$ the region corresponding to the $\mbox{BAC}(p,q)$ with $\frac{1}{2}<S(p,q)<1$ (the $3$-quasi-symmetric channels) represents the $47\%$ of the total area. The last percentage in each row of this table corresponds to the criterion $\mathcal{A}_{QS}^{n}$, associated with the $n$- quasi-symmetric channels. By (\ref{Equationtn}), the number of stable regions for $n=8$ and $n=9$ is $t_{8}=11$ and $t_{9}=14$ respectively. The percentages represented by these regions (from left to right) are $(9.09,\ 2.58,\ 3.85,\ 6.13,\ 10.54,\ 8.41,\ 12.12,\ 11.28,\ 7.00,\ 8.16,$ $20.84)$ for $n=8$ and $(7.28,\ 1.81,\ 2.58,\ 3.85,\ 6.13,\ 4.49,$    $6.06,\ 8.41,\ 12.12,\ 11.28,\ 7.00,\ 8.16,\ 4.59,\ 16.24)$, for $n=9$. As we can see, the criterion $\mathcal{A}_{QS}^{n}$ is the most probable criterion, when $(p,q)$ are chosen uniformly at random, among all the $n$-criteria for $4\leq n\leq 9$. We have also checked this fact for $n\leq 80$ and conjecture that this is true for any $n\geq 4$. Figure \ref{Dim40} displays graphically the percentages represented by all regions for $n = 40$. Let $\mathcal{A}_{Z}^{n}$ be the $n$-stable decision criterion for the BACs nearest to the criterion corresponding to the $Z$-channel. We also point out an interesting comparison regarding the sizes of the areas corresponding to the criteria $\mathcal{A}_{QS}^{n}$ and $\mathcal{A}_{Z}^{n}$. By considering for each $n$ the ratios $R(n)$ and $r(n)$ between the areas corresponding to $\mathcal{A}_{QS}^n$ and $\mathcal{A}_{Z}^{n}$ and the average area for this $n$, it can be observed from our data that $R(n)$ grows (with some very small  
oscillation) with $n$ linearly (i.e. $R(kn)/R(n)$ approaches $k$) whereas $r(n)$ gets near to one. As a sample, for $n= 4,\ 8,\ 16,\ 18,\ 25,\ 36,\ 49,\ 40,\ 50,\ 100$ and $200$, the obtained sequences $(n,R(n),r(n))$ are 
$(4,\ 1.418,\ 0.966),\ (8,\ 2.292,\ 0.100),\ (16,\ 3.908,\ 0.966),$ $(18,\ 4.398,\ 0.980),$ $(25,\ 5.867,\ 1.012),\ (36,\ 8.299,\ 0.978)$ $(40,\ 9.217,\ 0.983),\ (50,\ 11.588,\ 0.998),\ (100,\ 22.559,\ 0.991)$ and $(200,\ 45.098,\ 1.001)$.

\begin{figure}[h!]
	\begin{center}
		\includegraphics[scale=0.3]{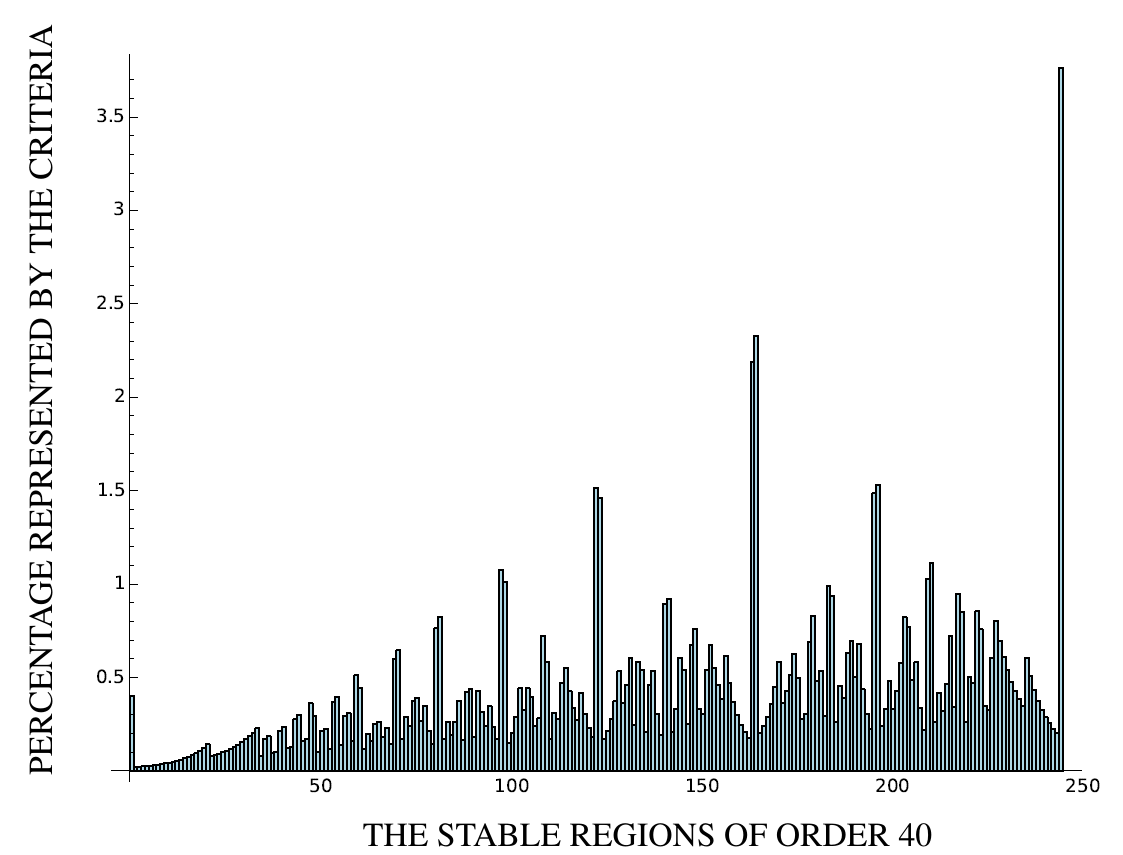}
	\end{center}
	\caption{The percentages represented by the $245$ stable regions for $n=40$ ordered by its BAC-function values. The rightmost bar corresponds to $\mathcal{A}_{QS}^{40}$.}
	\label{Dim40}
\end{figure}

\section*{Acknowledgment}

Work partially supported by CNPq grants 150270/2016-0, 140368/2015-9, 303985/2014-3 and 312926/2013-8 and by FAPESP grants 2015/26420-1 and 2013/25977-7.

\ifCLASSOPTIONcaptionsoff
  \newpage
\fi

\section*{Appendix}

{\sf Proof of the Proposition \ref{PropOrderedMatrixForm}.}

\noindent $i) \Rightarrow ii)$. We consider $x,y,z \in \X^{n}$ and the code $C=\{y,z\}$. Then, $\Pr_{W_1}(x|y)\leq \Pr_{W_1}(x|z)  \Leftrightarrow \argmax{c\in\{y,z\}} \Pr_{W_1}(x|c) \neq \{y\} 
 \Leftrightarrow \argmax{c\in\{y,z\}} \Pr_{W_2}(x|c) \neq \{y\}  \Leftrightarrow \Pr_{W_2}(x|y)\leq \Pr_{W_2}(x|z)$, where the second implication follows from assuming $i)$.
 
\noindent $ii) \Rightarrow i)$. We consider a code $C\subseteq \X^{n}$ and $x \in \X^{n}$, we have $c_0 \in \argmax{c\in C} \Pr_{W_1}(x|c) \Leftrightarrow  \Pr_{W_1}(x|c_0) \geq  \Pr_{W_1}(x|c),\ \forall c \in C   \Leftrightarrow  \Pr_{W_2}(x|c_0) \geq  \Pr_{W_2}(x|c),\ \forall c \in C   \Leftrightarrow  c_0 \in \argmax{c\in C} \Pr_{W_2}(x|c)$.

\noindent $ii) \Leftrightarrow iii)$. We consider the following equivalence relation in $\R^n$. Two vectors $a=(a_1,\ldots,a_n)$ and $a'=(a_1',\ldots,a_n')$ are equivalent (which is denoted by $a\sim a'$) if for all $i,j$: $a_i\leq a_j \Leftrightarrow a_i' \leq a_j'$. Let $\tau: \R^n \rightarrow \R^n$ be the function given by $\tau(a_1,\ldots,a_n)=(b_1,\ldots,b_n)$ with $b_i= \#\{j : a_j \leq a_i\}$, and $f_i$ and $f_i'$ be the $i$-th lines of $M_1$ and $M_2$, respectively. We note that $ii)$ is equivalent to $f_i \sim f'_i$, $\forall i: 1\leq i \leq 2^n$; and $iii)$ is equivalent to $\tau(f_i)=\tau(f'_i)$, $\forall i: 1 \leq i \leq 2^n$. Thus, it suffices to prove that for every $a=(a_1,\ldots,a_n)$ and $a'=(a'_1,\ldots, a'_n) \in \R^n$: $a\sim a' \Leftrightarrow \tau(a)=\tau(a')$. 

\noindent $(\Rightarrow)$ We consider $a,a' \in \R^n$ such that $a \sim a'$. For each $j: 1\leq j \leq n$ we have: $k\in\{i: a_i\leq a_j\} \Leftrightarrow a_k\leq a_j \Leftrightarrow a'_k\leq a'_j \Leftrightarrow  k \in \{i: a'_i\leq a'_j\}$. Thus $\{i: a_i\leq a_j\} = \{i: a'_i\leq a'_j\}$. Taking cardinalities we obtain $\tau(a)_j = \tau(a')_j$, $\forall j: 1\leq j \leq n$ and then $\tau(a)=\tau(a')$.

\noindent $(\Leftarrow)$  We consider $a,a' \in \R^n$ such that $\tau(a) = \tau(a')$ and note that if $\tau(a)_i \leq \tau(a)_j$ then $\{k: a_k \leq a_i\} \subseteq \{k: a_k \leq a_j\}$ (because in this case $\{k: a_k \leq a_j\} \not\subset \{k: a_k \leq a_i\}$). Thus, for all $i,j$ we have $a_i\leq a_j \Leftrightarrow \{k: a_k \leq a_i\} \subseteq \{k: a_k \leq a_j\} \Leftrightarrow \tau(a)_i \leq \tau(a)_j \Leftrightarrow \tau(a')_i \leq \tau(a')_j$ (because $\tau(a)=\tau(a')$) $\Leftrightarrow \{k: a'_k \leq a'_i\} \subseteq \{k: a'_k \leq a'_j\} \Leftrightarrow a'_i \leq a'_j$; and we conclude that $a\sim a'$. \\

{\sf Proof of the Theorem \ref{TheoChannelEqANDdecoders}.}

\noindent $i) \Rightarrow ii)$ Follows from the definition of $n$-equivalence and the fact that $\mbox{sdec}_{W_1}(C,x)=\mbox{sdec}_{W_2}(C,x)$ if and only if one of the following possibilities occurs:\\
\noindent $\bullet$ $\argmax{c\in C}\Pr_{W_1}(x|c)=\argmax{c\in C}\Pr_{W_2}(x|c)=\{c_0\}$ for some $c_0\in C$ or \\
\noindent $\bullet$ $\#\argmax{c\in C}\Pr_{W_1}(x|c)>1$ and $\#\argmax{c\in C}\Pr_{W_2}(x|c)>1$.\\

\noindent $ii) \Rightarrow i)$ We assume that $\mbox{sdec}_{W_1}=\mbox{sdec}_{W_2}$ for $n$-block codes, $n\geq 2$. This also implies that $\mbox{sdec}_{W_1}=\mbox{sdec}_{W_2}$ for $2$-block codes (see Remark \ref{RemarkN+1impliesN}). Let's suppose by contradiction that the channels $W_1$ and $W_2$ are not $n$-equivalent. In this case, there exists a code $C\subseteq \X^n$, $x\in \X^n$ and $c_0\in C$ such that $c_0 \in\argmax{c\in C}\Pr_{W_1}(x|c)$ and $c_0 \not\in \argmax{c\in C}\Pr_{W_2}(x|c)$. Thus, there exists $c_1\in C$ such that $\Pr_{W_2}(x|c_1)> \Pr_{W_2}(x|c_0)$. Consider the code $C'=\{c_0,c_1\}\subseteq C$. We have that $\argmax{c\in C'}\Pr_{W_2}(x|c)=\{c_1\}$. Since $\mbox{sdec}_{W_1}=\mbox{sdec}_{W_2}$ for $2$-block codes, we have that $\argmax{c\in C'}\Pr_{W_1}(x|c)=\{c_1\}$. Thus $\Pr_{W_1}(x|c_1)> \Pr_{W_1}(x|c_0)$ which contradicts that $c_0 \in\argmax{c\in C}\Pr_{W_1}(x|c)$. Thus, the channels $W_1$ and $W_2$ are $n$-equivalent.\\

\noindent $i) \Leftrightarrow iii)$ We define the extended probabilistic Voronoi regions $V_{(C,W)}^{\mbox{ext}}(c)=\{x\in \X^n : \Pr_{W}(x|c)\geq \Pr_{W}(x|c'), \forall c' \in C\}$. Using Proposition \ref{PropOrderedMatrixForm}, it is not difficult to prove that $W_1 \eqv W_2$ if and only if $V_{(C,W_1)}^{\mbox{ext}}(c)=V_{(C,W_2)}^{\mbox{ext}}(c)$ for all $C\subseteq \X^n$ and $c\in C$. By direct calculation we have 
\begin{equation}\label{Eq6}
\Pr(\mbox{udec}_{W}(C,x)=c)= \left\{\begin{array}{ll}
0 & \textrm{if $x \not\in V_{(C,W)}^{\mbox{ext}}(c)$};\\
\frac{1}{M} & \textrm{if $x\in V_{(C,W)}^{\mbox{ext}}(c)$};
\end{array}   \right.
\end{equation} where $M=\left|V_{(C,W)}^{\mbox{ext}}(c)\right|$. If $W_1\eqv W_2$, then $V_{(C,W_1)}^{\mbox{ext}}(c)=V_{(C,W_2)}^{\mbox{ext}}(c)$ for all $C\subseteq \X^n, x\in \X^n$ and by Equation (\ref{Eq6}) we have $\Pr(\mbox{udec}_{W_1}(C,x)=c)=\Pr(\mbox{udec}_{W_2}(C,x)=c)$ for all $C\subseteq \X^n, x\in \X^n$. This proves $i) \Rightarrow iii)$. The converse follows from the fact that if $\mbox{udec}_{W_1}= \mbox{udec}_{W_2}$ then 
\begin{align*}
V_{(C,W_1)}^{\mbox{ext}}(c) &= \{x\in \X^n: \Pr(\mbox{udec}_{W_1}(C,x)=c)>0\} \\
&= \{x\in \X^n: \Pr(\mbox{udec}_{W_2}(C,x)=c)>0\} \\
&= V_{(C,W_2)}^{\mbox{ext}}(c)
\end{align*}
(where in the first and last equality we use Equation (\ref{Eq6}) and $\mbox{udec}_{W_1}= \mbox{udec}_{W_2}$ is used in the second equality). From which we conclude $W_1\eqv W_2$.


\end{document}